\documentclass[12pt]{article}
\usepackage{amsmath}
\usepackage{amsfonts}
\usepackage{amssymb}
\usepackage{amsthm}
\usepackage{graphicx}
\usepackage{float}
\usepackage{tikz} 
\usepackage{xcolor}
\usepackage{natbib}
\usepackage{bm}
\usepackage{transparent}
\usepackage[pagewise]{lineno}
\usepackage{cancel}
\usepackage[affil-it]{authblk}
\newtheorem{propi}{Proposition}
\newtheorem{defi}{Definition}
\newtheorem{teo}{Theorem}
\newtheorem*{nonteo}{Theorem}

\newtheorem{coro}{Corollary}
\newtheorem{exa}{Example}
\newcommand{\B}{\boldsymbol}
\newcommand{\C}{\mathcal}
\title{\bf Fuzzy Classification Aggregation}
\author[1,2]{Federico Fioravanti\footnote{f.fioravanti@uva.nl\\ I am grateful to Fernando Tohm\'e, Ulle Endriss, Agust\'in Bonifacio, Jordi Mass\'o, participants of the COMSOC seminar at the ILLC, and two anonymous reviewers for comments and suggestions that led to an improvement of the paper.}} 
\affil[1]{GATE, Saint-Etienne School of Economics, Jean Monnet University, Saint-Etienne, France}
\affil[2]{Institute for Logic, Language and Computation, University of Amsterdam, Amsterdam, The Netherlands}
\date{\vspace{-5ex}}

\begin{document}

\maketitle
\begin{abstract}
We consider the problem where a set of individuals has to classify $m$ objects into $p$ categories and does so by aggregating the individual classifications.
We show that if $m\geq 3$, $m\geq p\geq 2$, and classifications are fuzzy, that is, objects belong to a category to a certain degree, then an optimal and independent aggregator rule that satisfies a weak unanimity condition belongs to the family of Weighted Arithmetic Means.
We also obtain characterization results for $m= p= 2$.

\noindent {\it Keywords}: Classification Aggregation; Weighted Arithmetic Mean; Fuzzy Setting.

\noindent {\it JEL Classification}: D71.

\end{abstract}
\section{Introduction}
The study of the problem of individuals classifying objects can be traced back to \citet{kasher1997question}.
In their paper, they consider a finite society that has to determine which one of its subsets of members consists of exactly those individuals that can be deemed to be part of a group named $J$. 
\citet{kasher1997question} gave rise to what can be regarded as a subdomain within social choice theory, namely, the study of the Group Identification Problem \citep[see, among others,][for more details]{samet2003between,miller2008group,fioravanti2021alternative}. 

A more general problem, where a group of individuals has to classify $m$ objects into $p$ different categories, has been considered by \citet{maniquet2016theorem}.
In their work, they study the case where there are at least as many objects as categories, at least three categories and all the categories must be filled with at least one object.\footnote{This surjectivity condition becomes relevant, for example, in situations where $m$ workers have to be assigned to $p$ tasks, and no task can be left unassigned.
It can also be seen as a sanity check of the classification process. 
For instance, consider a set of neural networks that is being trained to classify images into images of cats, dogs and rabbits. If a database with a large number of images of each of the three animals is used, then a neural network that fails to classify at least one image as a rabbit can be considered as not reliable.  }
This aspect is combined with three other seemingly natural properties, namely Unanimity, which states that when all individuals make the same classification then the aggregate classification complies; Independence, which is the requirement that an object be classified by the aggregator in the same way at any two classification profiles if every individual classifies it equally in both profiles; and Non-Dictatorship, which requires that there is no individual such that her classification is always selected.
The result is that there is no aggregator satisfying these three conditions.
Recently, \citet{cailloux2023classification} weaken the Unanimity axiom and find a weakening of the impossibility result that holds for $m>p\geq 2$, with the existence of an essential dictator.
This is an individual such that a permutation of her classification is always selected.
They also characterize the unanimous and independent aggregators when $m=p=2$. 

The surjectivity condition in \citet{maniquet2016theorem}, turns out to be a rather demanding requirement.
Without it, the impossibility result no longer holds, and many non-dictatorial aggregators satisfying Unanimity and Independence can be found in the literature \citep[see, for example,][]{kasher1997question}.

Our work considers the problem where $m$ objects have to be classified into $p$ categories, $m\geq p\geq 2$, and the classifications, both by the individuals and the rule, indicate degrees of membership of the objects to each of the categories.
Two possible scenarios where this setting may arise, are the following.
A national government has to allocate equal funds to all regions, to be used for different regional ministries (security, education, and health), with the central government providing some conditions on how the funds are distributed. 
It is natural to assume in this context, that the government sets a minimum amount of money that should be spent in each of the ministries, at least at a national level. 
Members of the government should decide on how the funds might be allocated for security, education, and health, ensuring that each category, across the country, receives at least the amount provided by the government to each region.
As an illustration, one region could be tasked with allocating half of the funds to security, a quarter to education, and a quarter to health. 
Another scenario involves assigning individuals to various mandatory tasks with equal time requirements, where proportions dictate how each person's time is allocated. 
For example, if tasks include laundry, lawn mowing, and grocery shopping, Bob, one of the individuals assigned to do the tasks, might spend half of his time on laundry, a quarter on mowing the lawn, and the remaining quarter on grocery shopping.\footnote{These examples illustrate why the notion of ``degree of membership'' can be considered natural to use, instead of considering probabilities. 
For instance, in the tasks allocation example, it is better if we know that Bob will spend with certainty half of his time doing the laundry, rather than knowing that there is a $50\%$ chance of Bob leaving all the clothes dirty.}


Fuzzy preferences, those that represent vagueness and uncertainty, are a useful tool that has been used in many aggregation problems. Noteworthy contributions include the work of \citet{DUTTA1987215}, who deals with exact choices under vague preferences, \citet{DUTTA198653} who investigate the structure of fuzzy aggregation rules determining fuzzy social orderings, and recent work by \citet{duddy2018some}, who prove the fuzzy counterpart of Weymark’s general oligarchy theorem \citep{weymark1984arrow}, and by \citet{raventos2020arrow}, who analyze \citeauthor{Arrow1951-ARRIVA}'s \citeyearpar{Arrow1951-ARRIVA} theorem in a fuzzy setting.

Numerous studies have addressed the Group Identification Problem when the preferences or classifications are not crisp.
For instance, \citet{CHO201866} present a model of group identification for more than two groups, allowing fractional classifications but no fractional opinions, \citet{ballester2008model} deal with fuzzy opinions in a sequential model, and \citet{fioravanti2022fuzzy} show that some of the impossibility results of \citet{kasher1997question} can be avoided. 
\citet{alcantud2019liberalism} analyze the classification problem in a fuzzy setting, and consider a strong fuzzy counterpart of the surjectivity condition, extending the impossibility result of \citet{maniquet2016theorem}.

We avoid both impossibility results, the crisp and the fuzzy ones, by considering a different, but also very natural, fuzzy surjectivity condition.
This condition, which in the crisp setting is extremely demanding, becomes a requirement that in the fuzzy setting allows the existence of a family of aggregators satisfying certain particular properties.
Thus, it is demanding enough to find a characterization result, but not so much to not be able to find one.
Building on some previous results in functional analysis by \citet{aczel1980characterizationwam,aczel1984aggregation}, and \citet{wagner1982allocation}, we show that the only optimal aggregators that are independent and that satisfy a weak version of Unanimity, are the Weighted Arithmetic Means (WAMs).
These rules aggregate individual inputs by computing a weighted average of their classifications, with weights reflecting the relative influence of each individual.

The Weighted Arithmetic Means are particularly significant for their theoretical and practical implications. 
Theoretically, they provide a framework that is both mathematically elegant and normatively appealing. 
By resolving prior impossibility results, they demonstrate that meaningful aggregation is achievable in the fuzzy setting under natural surjectivity conditions. 
Practically, their versatility and interpretability make them well-suited to a wide range of real-world decision-making problems.

In economic contexts, the Weighted Arithmetic Means offer a powerful tool for balancing efficiency and equity. 
For example, they can model scenarios where agents differ in expertise, authority, or stake, such as resource allocation in government planning, voting systems in elections, or task assignment in organisational settings. 
The weights assigned to individuals can reflect various factors, such as their expertise in evaluating specific categories, the resources they contribute, or their perceived reliability.

Consider the previously introduced example where a national government has to allocate equal funds to all regions, but now the members of the government are also representative of the different regions. 
A WAM-based approach allows to aggregate regional preferences into a collective allocation while respecting the relative importance or influence of each region. 
Similarly, in collaborative decision-making within organisations or project teams, WAMs facilitate the fair integration of expert opinions by assigning higher weights to individuals with specialized knowledge, ensuring that critical insights are not overlooked.

Furthermore, WAMs embody a natural interpretation of fairness: individuals with equal expertise or standing are treated equally, while those with greater expertise are appropriately weighted to reflect their contributions. 
This fairness property is particularly valuable in settings where the credibility or importance of agents varies, ensuring that outcomes remain justifiable and transparent.

The significance of these results lies not only in their theoretical robustness but also in their practical applicability to diverse fields, including economics, political science, and artificial intelligence, where fuzzy and nuanced classifications are increasingly common. 
By providing a clear axiomatic foundation, the Weighted Arithmetic Means offer a bridge between theory and application, enabling decision-makers to navigate complex classification problems in a principled and effective manner.


The plan of the paper is as follows.
Section \ref{basic} presents the basic notions and axioms that we use, while we present the results in Section \ref{results}. The Subsection \ref{sectionindependence} we include examples that show the independence of the axioms of our main result, Theorem~\ref{main}.
Finally, Section \ref{remarks} contains some concluding remarks.

\section{Basic Notions and Axioms}\label{basic}

Let $N=\{1,\ldots,n\}$ be a finite set of individuals and let $X=\{x_1,\ldots,x_m\}$ be a set of $m$ objects that need to be classified into the $p$ categories of a set $P$, with $p\geq 2$. 

The individuals classify each object according to a partial degree of membership to each category.
In the crisp setting, introduced by \citet{maniquet2016theorem}, classifications are surjective mappings $c:X\rightarrow P$, that is, every category must have at least one object classified into.
In the fuzzy setting, a {\it fuzzy classification} is a mapping $c:X\rightarrow [0,1]^P$ such that $\sum_{j=1}^mc(x_j)_t\geq 1$ for all $t\in P$ and $\sum_{t=1}^pc(x)_t=1$ for all $x\in X$. 
The former condition is the fuzzy counterpart of the surjectivity of the classification function, while the latter is the fuzzy counterpart of the assumption that every object $x$ must be assigned to exactly one category by each voter.
It is easy to see that these conditions imply that there must be at least as many objects as categories, thus $m\geq p$.

We use $\mathcal{C}$ to denote the set of fuzzy classifications, and every \mbox{$\boldsymbol{c}=(c_1,\ldots,c_n)\in \mathcal{C}^N$} is a {\it fuzzy classification profile}.
Given $\mathbf{c}\in\mathcal{C}^N$ and $x\in X$, we denote with $\mathbf{c}^x\in \C{C}^N|_x$ the {\em fuzzy classification profile restricted to x}, such that the entry $\mathbf{c}^x_{ij}$ indicates the degree of membership of the object $x$ to the category $j$, according to the agent $i$.
Thus, we have that $\boldsymbol{c}^x_{ij}=c_{i}(x)_j$.
This allows the representation of a fuzzy classification profile as a set of $m$ matrices in a subset of $[0,1]^{N\times P}$, one for each object, where row $i$ indicates the classifications given by the individual $i$, and column $t$ of the matrix indicates the classifications given by the set of individuals into the category $t$.

A {\it fuzzy classification aggregation function} (FCAF) is a mapping \mbox{$\alpha:\mathcal{C}^N\rightarrow \mathcal{C}$} such that $\alpha(\boldsymbol{c})(x)$ indicates the degrees of membership to the different categories of the object $x$.
We call the outcome of $\alpha$, the {\it fuzzy social classification}.

Next, we introduce a particular FCAF, the {\it Weighted Arithmetic Mean}. 
Let $\boldsymbol{w}=(w_1,\ldots,w_n)$ be a set of weights such that $w_i\in[0,1]$ for all $i\in N$ and $\sum_{i=1}^n w_i=1$.
Then \( \alpha_{\boldsymbol{w}}:\mathcal{C}^N\rightarrow\mathcal{C}\) is such that for all $x\in X$, $$\alpha_{\boldsymbol{w}}(\boldsymbol{c})(x)=w_{1}c_1(x)+\cdots+w_{n}c_n(x).$$
We say that the set of weights is {\em degenerate} if there is an $i\in N$ such that $w_i=1$.
If we think of the FCAF as a group of experts classifying objects, the WAM can be appropriate for situations where the individuals differ in expertise, with more experienced voters having higher weights.\footnote{Weighted aggregators have been previously studied in the social choice literature, in particular in the setting where monetary transfers are allowed \citep[see][]{roberts1979aggregation,MISHRA2012283}. See also \citet{stone1961opinionpooling} and \citet{dietrich2016probopinionpooling} for an use of this aggregator in opinion pooling.}
In particular, if \mbox{$w_i=\frac{1}{n}$} for all $i\in N$, we call this FCAF the Arithmetic Mean.

\begin{exa}
Consider a situation where there are $n=3$ individuals that have to classify $m=3$ objects into $p=3$ categories. 
Let $\boldsymbol{c}$ be a fuzzy classification profile such that:
\vspace{2mm}

$\boldsymbol{c}^{x_1}=\begin{pmatrix}
\frac{1}{3} & \frac{2}{3} & 0\\
1 & 0 & 0\\
\frac{1}{2} & \frac{1}{4} & \frac{1}{4}\\
\end{pmatrix}$,
$\boldsymbol{c}^{x_2}=\begin{pmatrix}
\frac{2}{3} & 0 & \frac{1}{3}\\
0 & 0 & 1\\
\frac{1}{2}  & 0 & \frac{1}{2}\\
\end{pmatrix}$, and
$\boldsymbol{c}^{x_3}=\begin{pmatrix}
0 & \frac{1}{3} & \frac{2}{3}\\
0 & 1 & 0\\
0 & \frac{3}{4} & \frac{1}{4} \\
\end{pmatrix}$.
\vspace{2mm}

In this example, individual $2$ classifies the object $x_2$ as being only part of category $3$ (second row of the second matrix).
If we consider the Weighted Arithmetic Mean with a non-degenerate set of weights $\boldsymbol{w}=(\frac{1}{2},0,\frac{1}{2})$, we obtain the following fuzzy social classification: 
\vspace{2mm}

\mbox{$\alpha_{\boldsymbol{w}}(\boldsymbol{c})(x_1)=(\frac{10}{24},\frac{11}{24},\frac{3}{24})$}, \mbox{$\alpha_{\boldsymbol{w}}(\boldsymbol{c})(x_2)=(\frac{14}{24},0,\frac{10}{24})$}, and \mbox{$\alpha_{\boldsymbol{w}}(\boldsymbol{c})(x_3)=(0,\frac{13}{24},\frac{11}{24})$}.
\end{exa}

In the following, we introduce a number of axioms, i.e., fundamental normative requirements that, depending on the specific classification problem at hand, a reasonable FCAF should satisfy.
The first axiom can be seen as an efficiency property.
It states that if all the individuals agree on how many objects should be classified into a given category, then the social classification must respect that.
\begin{defi}[Optimality]
An FCAF is optimal if for all $\boldsymbol{c}\in\mathcal{C}^N$ and all $t\in P$ such that $\sum_{j=1}^{m}c_i(x_j)_t=h$ for all $i\in N$, it is the case that $\sum_{j=1}^{m}\alpha(\boldsymbol{c})(x_j)_t=h$.  
\end{defi}
This axiom is useful in a setting where categories need to fulfil certain requirements and resources should not be wasted.
For example, in the case where a set of managers has to classify a group of workers as members of different departments within a company, in order to make them operative.
If all the managers consider that department A requires only two workers to be operative, then an optimal aggregator only classifies two workers into that department. 

The following axiom states that the fuzzy social classification of an object in two different fuzzy classification profiles does not change if the classification regarding that object is the same in both profiles for every individual. 
\begin{defi}[Independence]
An FCAF is independent if for all $\boldsymbol{c},\boldsymbol{c'}\in\mathcal{C}^N$ and all $x\in X$ such that \mbox{$c_i(x)=c'_i(x)$} for all $i\in N$, it is the case that \mbox{$\alpha(\boldsymbol{c})(x)=\alpha(\boldsymbol{c'})(x)$}.
 \end{defi}
This property is desirable in settings such as budget allocations, where the classification of one item (e.g., education spending) should not depend on unrelated factors (e.g., infrastructure spending).
Independence is a relatively weak requirement, as it asserts that the values of certain variables should be disregarded, in a setting where constraints are already imposed on the sums of these variables.
We can note that an independent FCAF $\alpha$ can be seen as a collection of mappings $(\alpha_x)_{x\in X}$, such that $\alpha_x:\C{C}^N|_x\rightarrow \C{C}|_x$ and $\alpha_x(\boldsymbol{c}^x)=\alpha(\boldsymbol{c})(x)$ \citep{cailloux2023classification}.
We call these mappings {\em Elementary} FCAFs.

A stronger version of Independence requires that if two different objects are classified equally in two different fuzzy classification profiles, then the fuzzy social classification must be the same for both objects.\footnote{This property is called Strong Label Neutrality in \citet{wagner1982allocation}.}
\begin{defi}[Symmetry]
 An FCAF is symmetric if for all $x,y\in X$ and all $\boldsymbol{c},\boldsymbol{c'}\in\mathcal{C}^N$ such that $c_i(x)=c'_i(y)$ for all $i\in N$, it is the case that \mbox{$\alpha(\boldsymbol{c})(x)=\alpha(\boldsymbol{c'})(y)$}.   
\end{defi}
Besides making the classifications of the rest of objects irrelevant when classifying, for example, an object $x_j$, a symmetric FCAF also ignores the names of the objects when classifying them.
In economic contexts, this can be seen as a fairness requirement, ensuring that identical cases are treated equally. 
For example, in labour market classifications, two identical workers should receive the same classification, regardless of arbitrary labels.
\citet{wagner1982allocation} shows that an optimal and symmetric FCAF has the same elementary FCAF for every object.

The next property states that if there is an object that is unanimously classified by the individuals, then the FCAF has to classify that object accordingly.
\begin{defi}[Unanimity]
An FCAF is unanimous if for all $\boldsymbol{c}\in\mathcal{C}^N$, all $x\in X$, all $t\in P$, and all $r\in[0,1]$ such that $c_1(x)_t=\cdots=c_n(x)_t=r$, it is the case that $\alpha(\boldsymbol{c})(x)_t=r$.    
\end{defi}
This property aligns with fundamental democratic principles in economic policy; if all stakeholders agree on a classification, it should be implemented.

A weaker version of Unanimity states that the only unanimous classification required to be respected is when all the individuals classify an object with a degree of membership $0$.
The intuition behind this axiom is that if all individuals consider that an object does not belong at all to a certain category, then the social classification should agree.
\begin{defi}[Zero Unanimity]
An FCAF is zero unanimous if for all \mbox{$\boldsymbol{c}\in\mathcal{C}^N$}, all $x\in X$, and all $t\in P$ such that \mbox{$c_1(x)_t=\cdots=c_n(x)_t=0$}, it is the case that $\alpha(\boldsymbol{c})(x)_t=0$.  
\end{defi}
This property can also be seen as a rationality constraint; if no member of the government wants to allocate funds for security in a specific region, then that region should not receive funds for that ministry.

Now we introduce our first fuzzy axiom, which can be seen as the fuzzy counterpart of Unanimity.\footnote{This axiom is called {\it coherence} by \citet{alcantud2019liberalism}. Here we use the same name used by \citet{fioravanti2022fuzzy}.}
It states that the degree to which an object is collectively classified must be between the classification degrees regarding that object for each of the individuals.
\begin{defi}[Fuzzy Consensus]
An FCAF satisfies Fuzzy Consensus if for all \mbox{$\boldsymbol{c}\in\mathcal{C}^N$}, all $x\in X$, and all $t\in P$, it is the case that $\alpha(\mathbf{c})(x)_t\in [\min_{i\in N}c_i(x)_t,\max_{i\in N}c_i(x)_t]$.    
\end{defi}
With an FCAF that satisfies Fuzzy Consensus, we can be confident that an object's classification remains close to the classifications assigned by the individuals, preventing extreme or arbitrary outcomes.

It is easy to see that Fuzzy Consensus implies Unanimity, and Unanimity implies Zero Unanimity.
These three axioms can be interpreted as the different `degrees' of consensus that we might require an aggregator to satisfy.

The next axiom states that there must not exist an individual that imposes her classification.
It is useful to prevent autocratic outcomes and ensure the participation of all the individuals.
\begin{defi}[Non-Dictatorship]
An FCAF is non-dictatorial if there is no individual $i\in N$ such that $\alpha(\mathbf{c})= c_i$ for all $\B{c}\in\mathcal{C}^N$.
\end{defi}
A stronger version of Non-Dictatorship states that the names of the individuals are not important for the aggregation of classifications.
This property ensures fairness and prevents undue influence based on identity rather than the merit of classifications.
\begin{defi}[Anonymity]
For a classification profile $\boldsymbol{c}\in \mathcal{C}^N$ and a permutation $\sigma:N\rightarrow N$ we define $\sigma(\boldsymbol{c})$ as the classification profile such that $\sigma(\boldsymbol{c})=\{\boldsymbol{c}_{\sigma(1)},\ldots,\boldsymbol{c}_{\sigma(n)}\}$.
An FCAF is anonymous if for all \mbox{$\boldsymbol{c}\in\mathcal{C}^N$} and all permutations $\sigma:N\rightarrow N$, it is the case that $\alpha(\sigma(\boldsymbol{c}))=\alpha(\boldsymbol{c})$.
\end{defi}

\section{Results}\label{results}
In the crisp setting, \citet{maniquet2016theorem} show the following impossibility result.
\begin{nonteo}{{\em \citep{maniquet2016theorem}}}
Let $m\geq p\geq3$.
There is no (non-fuzzy) Classification Aggregation Function that satisfies Independence, Unanimity, and Non-Dictatorship. 
\end{nonteo}
\citet{alcantud2019liberalism} extend this result to a fuzzy setting, by using Fuzzy Consensus and considering a stronger notion of surjectivity, where for each category there must exist an object with a degree of classification larger than $0.5$.
They show the existence of a fuzzy dictator, where whenever the dictator classifies an object into a category by more than $0.5$, then the object is classified into that category with a degree of more than $0.5$.
Thus, our main theorem can be seen as an escape from these impossibility results, even in a fuzzy setting.
\begin{teo}\label{main}
Let $m\geq 3$ and $m\geq p\geq 2$. 
An optimal Fuzzy Classification Aggregation Function satisfies Independence, Zero Unanimity, and Non-Dictatorship if, and only if, it is a Weighted Arithmetic Mean with a non-degenerate set of weights.
\end{teo}
Now we present a result by \citet{aczel1984aggregation} that is useful for our proof.
\begin{nonteo}{\em \citep{aczel1984aggregation}}
A family of mappings \mbox{$\{\alpha^{j}:[0,1]^N\rightarrow[0,1]\}_{j=1}^m$} satisfies Zero Unanimity and the following condition:
\begin{equation}\label{optimality}
\text{if } z^j\in [0,1]^N \text{ and }  \sum_{j=1}^m z^j=(1,\ldots,1), \text{ then } \sum_{j=1}^m\alpha^{j}(z^j)=1    
\end{equation}
if, and only if, there exists a set of weights $(w_1,\ldots,w_n)$ such that for each $j=1,\ldots,m$ and any $z\in[0,1]^N$, we have that $$\alpha^{j}(z)=w_1z_{1}+\cdots + w_nz_{n}.$$
 
 \end{nonteo} 
\begin{proof}[Proof of Theorem~\ref{main}]
That a Weighted Arithmetic Mean with a non-degenerate set of weights satisfies the four axioms is straightforward to see.
For the only if part, we use the result by \citet{aczel1984aggregation}. 
Let $\alpha$ be an optimal FCAF that satisfies Independence and Zero Unanimity.
Recall that an independent FCAF can be seen as a collection of elementary FCAFs.
So we can consider a family of mappings \mbox{$\{\alpha_{x_j}:\C{C}^N|_{x_j}\rightarrow \C{C}|_{x_j}\}_{j=1}^m$}.
By optimality, we obtain that if $\sum_{j=1}^{m}c_i(x_j)_t=h$ for all $i\in N$, it is the case that $\sum_{j=1}^{m}\alpha_{x_j}(\boldsymbol{c}^{x_j})_t=h$.
Then we have that for each $t\in P$ there is a family of mappings $\{\alpha_{x_j}|_t:\C{C}^N|_{x_j\land t}\rightarrow[0,1]\}_{j=1}^m$ satisfying condition (\ref{optimality}).\footnote{$\C{C}^N|_{x_j\land t}$ is the set of classifications restricted to the object $x_j$ and the category $t$.
That is, an element of $\C{C}^N|_{x_j\land t}$ is the column $t$ of a matrix in $\C{C}^N|_{x_j}$.} 
The function $\alpha_{x_j}|_t$ maps the column $t$ in $\boldsymbol{c}^{x_j}$ and outputs the degree of classification of object $x_j$ into the category $t$.
Due to optimality, regardless of the values that are not in column $t$, we can be sure that condition (\ref{optimality}) is satisfied by $\alpha_{x_j}|_t$ and that the degree of classification of the object $x_j$ into the category $t$ depends on a weighted aggregator.
Thus we have that for each category $t$ there is a set of weights $(w^t_1,\ldots,w^t_n)$ such that for $j=1,\ldots,m$, it is the case that
$$\alpha_{x_j}|_t(\boldsymbol{c}^{x_j}_{1t},\ldots,\boldsymbol{c}^{x_j}_{nt})=w^t_1\boldsymbol{c}^{x_j}_{1t}+\cdots+w^t_n\boldsymbol{c}^{x_j}_{nt}.$$
What is left to show is that for every $i\in N$, the weights are the same for every category, meaning that $w_i^t=w_i^{t'}$ for all $t,t'\in P$.

Fix an $i^\star\in N$ and consider a set of classification problems $\{\B{c}^t\}_{1\leq t\leq p}$ such that for each $\B{c}^t$, it is the case that $c^t_j(x_1)=(1,0,\ldots,0)$ for all $j\neq i^\star$ and $c^t_{i^\star}(x_1)=(0,\ldots,0,\underset{t-\text{th}}{1},0,\ldots,0)$.
Thus, as for each classification problem $\B{c}^t$ we have that $\alpha(\B{c}^t)(x_1)=(\sum_{j\neq i^\star}w_j^1,\ldots,0,\underset{t-\text{th}}{w_{i^\star}^t},0,\ldots,0)$, we obtain the constraint
$$\sum_{j\neq i^\star}w_j^1+w_{i^\star}^t=1\hspace{20mm}(t)$$  
In consequence, for any pair of classification problems $\B{c}^t,\B{c}^{t'}$, we can substract the constrains $(t)$ and $(t')$ and obtain that $w_{i^\star}^t=w_{i^\star}^{t'}$. 
This conclusion holds for all $t,t'\in P$, and for each $i\in N$.

So the FCAF associates to every individual $i\in N$ a weight $w_i$ such that $w_i\geq 0$ and that $\sum_{i=1}^nw_i=1$, and as it is non-dictatorial, we have that there is no $i\in N $ such that $w_i=1$.
Thus the FCAF is a Weighted Arithmetic Mean with a non-degenerate set of weights.
\end{proof}
The impossibility result from the crisp setting is obtained when the aggregator is a Weighted Arithmetic Mean and the set of weights is degenerate. 
As the aggregator must have a crisp outcome, this is the only possible set of weights (except for permutations on the name of the individuals).
One notable aspect of this result is that the weight assigned to each individual remains consistent across all categories, with no category receiving a higher weight than the others.
An implication of our main result is that if we want all individuals to be considered the same, then the weights must be the same for all individuals.
\begin{coro}
Let $m\geq 3$ and $m\geq p\geq 2$. 
An optimal Fuzzy Classification Aggregation Function satisfies Independence, Zero Unanimity, and Anonymity if, and only if, it is the Arithmetic Mean.   
\end{coro}
As the Weighted Arithmetic Means with non-degenerate weights satisfy Unanimity and Fuzzy Consensus, stronger versions of Zero Unanimity, we have the following two corollaries.
\begin{coro}
Let $m\geq 3$ and $m\geq p\geq 2$. 
An optimal Fuzzy Classification Aggregation Function satisfies Independence, Unanimity, and Non-Dictatorship if, and only if, it is a Weighted Arithmetic Mean with a non-degenerate set of weights.
\end{coro}
\begin{coro}
Let $m\geq 3$ and $m\geq p\geq 2$. 
An optimal Fuzzy Classification Aggregation Function satisfies Independence, Fuzzy Consensus, and Non-Dictatorship if, and only if, it is a Weighted Arithmetic Mean with a non-degenerate set of weights.
\end{coro}
The impossibility results established by \citet{maniquet2016theorem} and \citet{alcantud2019liberalism} are successfully circumvented in our framework, even while preserving their original `consensual' axioms.
This breakthrough is primarily due to our refined and less restrictive definition of a surjective category, which serves as the foundation for our results.
By weakening the surjectivity conditions proposed in these earlier works in a natural and intuitive manner, we demonstrate that properties like Independence and Unanimity, which are often regarded as overly demanding, become more attainable in the context of fuzzy classifications.

We emphasize that Independence does not impose any condition between objects, in the sense that we can use different elementary FCAFs for each object.
But the combination of Optimality, Independence, and Zero Unanimity$\backslash$Unanimity$\backslash$Fuzzy Consensus, in this setting with its particular surjectivity conditions, forces the FCAF to also satisfy the symmetry condition, thus obtaining the following corollary.
\begin{coro}
Let $m\geq 3$ and $m\geq p\geq 2$.
An optimal and symmetric Fuzzy Classification Aggregation Function satisfies Zero Unanimity$\backslash$Unanimity$\backslash$Fuzzy Consensus, and Non-Dictatorship if, and only if, it is a Weighted Arithmetic Mean with a non-degenerate set of weights.
\end{coro}

It is worth mentioning that the previous results are only valid for $m\geq3$.
For the case of $m=2$ objects and $p=2$ categories, Independence is trivially satisfied (due to the surjectivity conditions).
This forces us to use a stronger axiom to obtain a characterization.

For optimal, symmetric and zero unanimous FCAFs, we can have FCAFs that are not Weighted Arithmetic Means, as the following result shows.
\begin{propi}\label{m=2ZU}
Let $m=p=2$.
An optimal and symmetric Fuzzy Classification Aggregation Function $\alpha$ satisfies Zero Unanimity if, and only if, there is a function $h:[-\frac{1}{2},\frac{1}{2}]^n\rightarrow[-\frac{1}{2},\frac{1}{2}]$ where 
\begin{equation*}\label{odd}
    h(x_1,\ldots,x_n)=-h(-x_1,\ldots,-x_n),
\end{equation*}
and
\begin{equation*}\label{hunanimous}
    h\left(\dfrac{1}{2},\ldots,\dfrac{1}{2}\right)=\dfrac{1}{2},
\end{equation*}
such that for each $\boldsymbol{c}\in\mathcal{C}^N$, $$\alpha (\boldsymbol{c})(x_t)=\left(h\left(\boldsymbol{c}^{x_t}_{11}-\frac{1}{2},\ldots,\boldsymbol{c}^{x_t}_{n1}-\frac{1}{2}\right)+\frac{1}{2},h\left(\boldsymbol{c}^{x_t}_{12}-\frac{1}{2},\ldots,\boldsymbol{c}^{x_t}_{n2}-\frac{1}{2}\right)+\frac{1}{2}\right)$$ for $t=1,2$.
\end{propi}
Before the proof, we present the following result by \citet{wagner1982allocation}.
\begin{nonteo}\label{star}{\em \citep{wagner1982allocation}}
Let $m=p=2$.
An optimal and symmetric mapping \mbox{$\alpha:[0,1]^{N\times 2}\rightarrow[0,1]^2$} satisfies Zero Unanimity if, and only if, there exists a function $h:[-\frac{1}{2},\frac{1}{2}]^n\rightarrow[-\frac{1}{2},\frac{1}{2}]$ where 
\begin{equation*}
    h(y_1,\ldots,y_n)=-h(-y_1,\ldots,-y_n),
\end{equation*}
and
\begin{equation*}
    h\left(\dfrac{1}{2},\ldots,\dfrac{1}{2}\right)=\dfrac{1}{2},
\end{equation*}    
such that for each $\B{c}\in [0,1]^{N\times 2}$, we have that $$\alpha(\B{c})=\left(h\left(\B{c}_{11}-\frac{1}{2},\ldots,\B{c}_{n1}-\frac{1}{2}\right)+\frac{1}{2},h\left(\B{c}_{12}-\frac{1}{2},\ldots,\B{c}_{n2}-\frac{1}{2}\right)+\frac{1}{2}\right).$$
\end{nonteo}
\begin{proof}[Proof of Proposition~\ref{m=2ZU}.]
That such an FCAF satisfies Zero Unanimity is straightforward to see.
For the if part, we use the result by \citet{wagner1982allocation}.
By Optimality and Zero Unanimity, for each category the hypotheses of \citet{wagner1982allocation} are satisfied and thus, we have a mapping $h$ satisfying \citeauthor{wagner1982allocation}'s conditions, namely $h^1$ and $h^2$.
What is left to see is that $h^1=h^2$.
We have the following conditions that must be satisfied by $\alpha (\boldsymbol{c})(x_1)$ and $\alpha (\boldsymbol{c})(x_2)$:
$$h^1\left(\boldsymbol{c}^{x_1}_{11}-\frac{1}{2},\ldots,\boldsymbol{c}^{x_1}_{n1}-\frac{1}{2}\right)+\frac{1}{2}+h^2\left(\boldsymbol{c}^{x_1}_{12}-\frac{1}{2},\ldots,\boldsymbol{c}^{x_1}_{n2}-\frac{1}{2}\right)+\frac{1}{2}=1,$$
and 
$$h^1\left(\boldsymbol{c}^{x_1}_{11}-\frac{1}{2},\ldots,\boldsymbol{c}^{x_1}_{n1}-\frac{1}{2}\right)+\frac{1}{2}+h^1\left(\boldsymbol{c}^{x_2}_{11}-\frac{1}{2},\ldots,\boldsymbol{c}^{x_2}_{n1}-\frac{1}{2}\right)+\frac{1}{2}=1.$$
It is easy to see that $\boldsymbol{c}^{x_1}_{j2}=\boldsymbol{c}^{x_2}_{j1}$ for all $i=1,\ldots,n$.
Thus we obtain that $h^1\left(\boldsymbol{c}^{x_2}_{11}-\frac{1}{2},\ldots,\boldsymbol{c}^{x_2}_{n1}-\frac{1}{2}\right)=h^2(\boldsymbol{c}^{x_2}_{11}-\frac{1}{2},\ldots,\boldsymbol{c}^{x_2}_{n1}-\frac{1}{2})$, concluding our proof.

\end{proof}
Despite its technical flavour, Proposition~\ref{m=2ZU} has as a particular case the class of Weighted Arithmetic Means, and also any other power mean with an odd exponent.\footnote{A power mean with an odd exponent is a function $f:\mathbb{R}^n\rightarrow\mathbb{R}$ such that $f(x_1,\ldots,x_n)=(\frac{1}{n}\sum_{i=1}^{n}x_i^m)^{\frac{1}{m}}$ where $m$ is an odd number.}
Proposition $1$ of \citet{cailloux2023classification} is a particular case of Proposition~\ref{m=2ZU}, where condition~\ref{odd} is implied by the complementary condition of their result.

A more general result, valid for $m\geq p\geq 2$ can be attained with some slightly changes in the model.
For this general case, let $A\subseteq X$, $|A|=a\geq 2$, and consider fuzzy classifications $c^\ast_A:A\rightarrow \mathbb{R}^P$, where the surjectivity requirements are such that for all $x\in A$, and for a given $s\in\mathbb{R}$, it is the case that $\sum_{t=1}^pc_A^\ast(x)_t=s$, and that $\sum_{j=1}^ac_A^\ast(x_j)_t\geq s$ if $s\geq 0$ for all $t\in P$, or $\sum_{j=1}^ac_A^\ast(x_j)_t< s$ if $s< 0$ for all $t\in P$.
We can think of this setting as $n$ operators assigning $s$ hours of use of at most $m$ machines into $p$ different tasks that need $s$ hours to be finished, or $s$ euros that at most $m$ persons have to spend in $p$ different projects that need at least $s$ euros to be done.
This setting can be relevant when not always all the machines or all the persons are available, but we need the decision process to be consistent across the different possibilities.
For each $A\subseteq X$ we use $\mathcal{C}_A^\ast$ to denote this set of classifications.
We define an FCAF$^\ast$ $\alpha$ as a set of FCAFs over subsets of $X$, that is, $\alpha=\{\alpha_A\}_{A\subseteq X}$, where $\alpha_A:{\mathcal{C}^{\ast}_A}^N\rightarrow \mathcal{C}^{\ast}_A$. 
We say that an FCAF$^\ast$ is optimal if $\alpha_A$ is optimal for all $A\subseteq X$.
We similarly define a symmetric, fuzzy consensual and non-dictatorial FCAF$^\star$.

We have that an optimal FCAF$^\ast$ satisfies the {\it $k$-allocation} property for any $k=a\geq2$ \citep{aczel1980characterizationwam}, that is, if for any $t\in P$ such that $\sum_{j=1}^ac^\ast_A(x_j)_t= s$ for all $i\in N$, it is the case that \mbox{$\sum_{j=1}^{m}\alpha(\boldsymbol{c}_A^\ast)(x_j)_t=s$}. 
The $k$-allocation property allows us to characterize the optimal and symmetric FCAF$^\ast$s that satisfy Fuzzy Consensus for all values of $m$ such that $m\geq p\geq 2$.
\begin{teo}\label{second}
 Let $m\geq p\geq 2$.
 An optimal and symmetric FCAF$^{\ast}$ satisfies Fuzzy Consensus and Non-Dictatorship if, and only if, it is a Weighted Arithmetic Mean with a non-degenerate set of weights.   
\end{teo}
\begin{proof}
\citet{aczel1980characterizationwam} show that a mapping $\boldsymbol{c}:\mathbb{R}^n\rightarrow\mathbb{R}$ satisfies the k-allocation property for $k=2,3$ and is bounded if, and only if, it is a Weighted Arithmetic Mean.
In our setting, symmetry implies that for a given category $t\in P$ and a given $A\subseteq X$, $\alpha_A$ can be seen as a set of elementary FCAFs $\{\alpha_x\}_{x\in A}$, where $\alpha_x=\alpha_y$ for all $x,y\in A$.
By optimality, the k-allocation property is satisfied by $\alpha_A$ for all $A\subseteq X$, in particular for the cases where $|A|=2$ or $|A|=3$.
Fuzzy Consensus implies they are bounded and thus, by \citet{aczel1980characterizationwam}, the mapping $\alpha_A$ is a Weighted Arithmetic Mean for all $A\subseteq X$.
Finally, Non-Dictatorship implies that the set of weights is non-degenerate.
By a similar proof to the one used in Theorem~\ref{main}, it is easy to see that the mappings $\alpha_A$ are the same for every category, concluding our proof.
\end{proof}
\subsection{Independence of the Axioms in Theorem~\ref{main}}\label{sectionindependence}
We conclude this section by showing that the axioms in our main result are independent. 
For each of the axioms in the statement of Theorem \ref{main}, we exhibit an example of an FCAF, different from a Weighted Arithmetic Mean, that satisfies all the axioms but one.
\begin{itemize}   
    \item \emph{All but Optimality:} let $m=3$ and $p=2$. 
    The FCAF $\alpha$ such that for all $\B{c}\in\mathcal{C}^N$, for any $x\in X$ such that $c_1(x)=(1,0)$ or $c_1(x)=(0,1)$, then $\alpha(\B{c})(x)=c_1(x)$; otherwise it is $\alpha(\B{c})(x)=(\frac{1}{2},\frac{1}{2})$.
    This FCAF is almost always constant, except when the individual $1$ proposes extreme classifications.
    
    \item \emph{All but Independence:} the FCAF $\alpha$ such that if $\B{c}\in \mathcal{C}^N$ is such that for all $i\in N$, $c_i(x_1)=e_i$ (the canonical vector), then $\alpha(\B{c})=c_1$; otherwise it is $\alpha(\B{c})=c_2$.
     \item \emph{All but Zero Unanimity:} the FCAF $\alpha$ such that for all $\B{c}\in \mathcal{C}^N$, $\alpha(\B{c})(x_1)=c_1(x_2)$,  $\alpha(\B{c})(x_2)=c_1(x_1)$, and $\alpha(\B{c})(x_j)=c_1(x_j)$ for \mbox{$j\neq 1,2$}.
     This FCAF is essentially a dictatorship \citep{cailloux2023classification}.

    \item \emph{All but Non-Dictatorship:} the FCAF such that for all $\B{c}\in\mathcal{C}^N$, it is the case that $\alpha(\mathbf{c})= c_i$ for all $\B{c}\in\mathcal{C}^N$, for a fixed $i\in N$.
\end{itemize}
\section{Final Remarks}\label{remarks}
In this work, we present an analysis from a fuzzy point of view, of the challenge of classifying $m$ objects into $p$ different categories.
The classifications of the objects by the agents and the rules are no longer crisp statements about to which category the object is assigned. 
Instead of that, the classifications are expressed in terms of degrees of assignment to each of the categories.
In the crisp setting for more than two objects and two categories, requiring the rule to fill each category with at least one object, to be independent and unanimous, implies the existence of an agent such that objects are classified according to the opinions of that agent \citep{maniquet2016theorem}.
Even with a weaker unanimity condition, this result can not have a major improvement \citep[][showing the existence of an essential dictator]{cailloux2023classification}.

The fuzzy setting proves advantageous, as the surjectivity conditions can assume diverse yet natural interpretations. 
Strong interpretations, akin to those in \cite{alcantud2019liberalism}, can extend the scope of impossibility results. 
However, by considering weaker surjectivity conditions and different versions of Unanimity, we can circumvent these limitations. 
Our findings show that rules that belong to the family of Weighted Arithmetic Means are the only ones that satisfy Optimality, Independence, and Zero Unanimity.
Under different interpretations of the classification process and the consensual axioms, we can obtain characterization results that hold for $m\geq p\geq 2$.

The significance of these results extends beyond their theoretical contributions to the field of social choice and aggregation theory. 
The Weighted Arithmetic Means provide a flexible and principled framework for addressing classification problems in fuzzy settings. 
Their applicability spans fields like economics, political science, and artificial intelligence, where fairness, consistency, and independence are essential. 
By bridging theoretical principles and practical decision-making, these rules enable equitable and efficient aggregation in complex real-world scenarios.

An intriguing avenue for further inquiry lies in extending our analysis to dynamic scenarios, where object classifications may evolve over time. 
Exploring the temporal dynamics of fuzzy classifications could provide valuable insights into the adaptability and stability of the proposed framework in real-world applications.\\

\textbf{Declarations of interest:} None.

\bibliography{ref}

\end{document}